\documentclass{pspum-l}

\newtheorem{theorem}{Theorem}[section]
\newtheorem{lemma}[theorem]{Lemma}

\theoremstyle{definition}

\theoremstyle{remark}
\newtheorem{remark}[theorem]{Remark}

\numberwithin{equation}{section}

\begin{document}

\title[A Minimal Uncertainty Product for Semiclassical Wave Packets]%
{A Minimal Uncertainty Product for
One Dimensional Semiclassical Wave Packets}

%\runhead{A Minimal Uncertainty Product for Semiclassical Wave Packets}

\author{George A. Hagedorn}
\address{Department of Mathematics and Center for Statistical Mechanics,
Mathematical Physics, and Theoretical Chemistry, Virginia Polytechnic Institute
and State University, Blacksburg, Virginia 24061--0123}
\email{hagedorn@math.vt.edu}
\thanks{The author was supported in part by NSF Grant \#DMS--1210928.}

%    General info
\date{July 14, 2012.}

\dedicatory{Happy $60^{\mbox{\scriptsize th}}$ birthday, Mr.~Fritz!}
%{Dedicated to the celebration of Fritz Gesztesy's 60th birthday!}


\begin{abstract}
Although real, normalized Gaussian wave packets minimize the product
of position and momentum uncertainties, generic complex
normalized Gaussian wave packets do not.
We prove they minimize an alternative product of uncertainties that correspond to variables that are phase space rotations of position and momentum.   
\end{abstract}

\maketitle

\section{Introduction}
In studying small $\hbar$ asymptotics of solutions to the time--dependent
Schr\"odinger equation, semiclassical wave packets have proven very useful.
(See {\it e.g.}, \cite{raise} and its references, or more modern works, such as
\cite{FGL}, \cite{GH}, or \cite{HJ}.)
In the case of one degree of freedom,
for fixed, allowed values of the parameters
$A,\,B,\,\hbar,\,a,\,\mbox{and}\,\eta$, 
these wave packets are an orthonormal basis of
$L^2(\mathbb{R},\,dx)$
that we denote by
$\{\,\varphi_k(A,\,B,\,\hbar,\,a,\,\eta,\,x)\,\}$, where $k=0,\,1,\,\ldots$\,.

By making a proper choice of the parameters, one can write any normalized,
one dimensional complex Gaussian wave packet as 
$\varphi_0(A,\,B,\,\hbar,\,a,\,\eta,\,x)$.
Two of the restrictions on the parameters are the conditions
$\mbox{Re}\,\overline{A}B\,=\,1$ and $\hbar>0$.
By a fairly straightforward calculation, one can prove that
the usual position, momentum uncertainty product in the state
$\varphi_0(A,\,B,\,\hbar,\,a,\,\eta,\,x)$ satisfies
$\displaystyle\Delta x\,\Delta p\,=\,\frac{\hbar}2\,|A|\,|B|$.
When the complex phases of $A$ and $B$ are the same,
$\varphi_0(A,\,B,\,\hbar,\,a,\,\eta,\,x)$ is a phase times a real
Gaussian, and it is a standard result that this product takes its minimal value 
$\displaystyle\Delta x\,\Delta p\,=\,\frac{\hbar}2$.
When the phases of $A$ and $B$ are different, $|A|\,|B|>1$,
and consequently,
$\displaystyle\Delta x\,\Delta p\,>\,\frac{\hbar}2$.
As usual, $x$ and $p$ are the position and momentum operators, and
we define the uncertainty for a self-adjoint observable $X$
in the state $\psi$ to be
$\Delta X\,=\,\sqrt{\langle\,\psi,\,X^2\,\psi\,\rangle\,-\,
\langle\,\psi,\,X\,\psi\,\rangle^2}$.

The goal of this paper is to prove an alternative minimum uncertainty product
result for the state
$\varphi_0(A,\,B,\,\hbar,\,a,\,\eta,\,x)$.
We define two ``rotated'' operators
\begin{eqnarray}\nonumber
\alpha&=&\phantom{-}\,\cos(\theta)\ x\ +\ \sin(\theta)\ p
\qquad\qquad\mbox{and}
\\[2mm]\nonumber
\beta&=& -\,\sin(\theta)\ x\ +\ \cos(\theta)\ p.
\end{eqnarray}
We show that in any normalized state,
$$
\displaystyle\Delta\alpha\,\Delta\beta\,\ge\,\frac{\hbar}2,
$$
and we show in Theorem \ref{thm:main} that by choosing
$$
\theta\,=\,\frac 12\,\mbox{arctan}\,
\left(\,\frac{2\ \mbox{Im}\,(B\overline{A})}{|B|^2-|A|^2}\,\right),
$$
one has
$$
\Delta\alpha\,\Delta\beta\,=\,\frac{\hbar}2
$$ 
in the state
$\varphi_0(A,\,B,\,\hbar,\,a,\,\eta,\,x)$.
So, general semiclassical wave packets satisfy a minimal uncertainty relation.
It is not the usual relation, but the product for the rotated operators.

\begin{remark}Employing a sort of microlocal intuition,
we often like to think heuristically of supports of
quantum states in phase space. Weyl asymptotics  and
Bohr--Sommerfeld rules suggest that
a normalized state should occupy a phase space area of $2\pi\hbar$.
From this viewpoint, we envision the usual, frequency $\omega$
harmonic oscillator ground state as having phase space support of
the interior of an ellipse that is centered at the origin and has semiaxes
$\sqrt{2\,\hbar/\omega}$ in the $x$ direction and
$\sqrt{2\,\hbar\,\omega}$ in the $p$ direction.

On this intuitive level, the region of phase space corresponding to\\
$\varphi_0(A,\,B,\,\hbar,\,a,\,\eta,\,x)$ is such an ellipse that has been
rotated through the angle $\theta$
and then translated so its center is at $(a,\,\eta)$.
\end{remark}

\begin{remark}The situation in more than one dimension is significantly
more complicated.  Although we believe an analogous result must be true,
and we have received some preliminary calculations from Vidian Rousse
\cite{rousse}, we have not seen a complete argument.
In $n$ dimensions, one would like to find an
orthogonal, symplectic matrix depending on
 $A$ and $B$
that would define rotated operators
$\alpha_j$ and $\beta_j$ for $j=1,\,2,\,\cdots,\,n$,
such that in the state
$\varphi_0(A,\,B,\,\hbar,\,a,\,\eta,\,x)$,
$$
\Delta\alpha_j\,\Delta\beta_j\ =\ \frac\hbar 2
$$
for each $j$.

We have tried to generalize the one dimensional proof given in
Section \ref{MainSection} to $n$ dimensions. The proof relies on
finding the minimum of a function of $n^2$ variables.
(The symplectic orthogonal group is isomorphic to the
$n^2$ dimensional real Lie group $U(n)$.)
Using the second derivative test to separate minima,
maxima, and saddles becomes exceptionally complicated. 
\end{remark}

\vskip 5mm
%%%%%%%%%%%%%%%%%%%%
% Section 2
%%%%%%%%%%%%%%%%%%%%
\section{Observations about the Angle of Rotation}\label{angle}
The particular value of $\theta$ may seem rather bizarre, but it is natural.
The wave packet $\varphi_0(A,\,B,\,\hbar,\,a,\,\eta,\,x)$ is the ground
state of a Hamiltonian that is quadratic in $x$ and $p$. That Hamiltonian
is explicitly
$$
H\ =\ \frac 12\ (\,x\quad p\,)\,\left(
\begin{array}{cc}|B|^2&\mbox{Im}\,(B\overline{A})\\[2mm]
\mbox{Im}\,(B\overline{A})&|A|^2\end{array}\right)\,
\left(\begin{array}{c}x\\[2mm] p\end{array}\right).
$$
In terms of the raising and lowering operators we use below,
$$
H\ =\ \frac \hbar2\ \left(\mathcal{A}^*\mathcal{A}\,+\,
\mathcal{A}\mathcal{A}^*\right).
$$
A rotation through angle $\theta$ diagonalizes the real symmetric matrix
$$
\frac 12\
\left(
\begin{array}{cc}|B|^2&\mbox{Im}\,(B\overline{A})\\[2mm]
\mbox{Im}\,(B\overline{A})&|A|^2\end{array}\right),
$$
whose eigenvalues are
$$
\frac 14\ \left\{\,(|A|^2+|B|^2)\ \pm\
\sqrt{(|A|^2-|B|^2)^2\,+\,4\,(\mbox{Im}\,(B\overline{A}))^2}\ \right\}.
$$
The product of these eigenvalues is the determinant of the matrix
for $H$. The calculation is a bit tedious, but for allowed values of
the parameters, the product is $1/4$.

When $H$ is associated with a diagonal matrix, $\theta=0$, and 
one is essentially back to considering a standard frequency $\omega$
harmonic oscillator Hamiltonian, where $A=\omega^{-1/2}$
and $B=\omega^{1/2}$.  In this case, the standard uncertainty product
is $\hbar/2$. When considering the matrix associated with $H$,
the value of $\theta$ is very natural.

The vectors
$\varphi_k(A,\,B,\,\hbar,\,a,\,\eta,\,x)$ also diagonalize the Hamiltonians
$$
H_1\ =\ \hbar\,\mathcal{A}\mathcal{A}^*\qquad\mbox{and}\qquad
H_2\ =\ \hbar\,\mathcal{A}^*\mathcal{A}.
$$
The matrices associated with these quadratic Hamiltonians are
$$
\frac 12\ \left(
\begin{array}{cc}|B|^2&-\,i\,B\overline{A}\\[2mm]
i\,A\overline{B}&|A|^2\end{array}\right)
\qquad\mbox{and}\qquad
\frac 12\ \left(
\begin{array}{cc}|B|^2&i\,A\overline{B}\\[2mm]
-\,i\,B\overline{A}&|A|^2\end{array}\right),
$$
respectively. In general, these matrices 
are not real symmetric and do not have real eigenvectors.
However, the quadratic forms in the classical variables $x$ and $p$
are the same as the one for $H$.
So, they lead to the same value of $\theta$.

\vskip 5mm
%%%%%%%%%%%%%%%%%%%%
% Section 3
%%%%%%%%%%%%%%%%%%%%
\section{Preliminary Comments about Uncertainty Products}
We wish to begin with the standard argument for the Heisenberg uncertainty
relation
\begin{equation}\label{Heisenberg}
\Delta x\,\Delta p\ \ge\ \frac{\hbar}2.
\end{equation}
This is a consequence of the more general result that for any two
self-adjoint operators
$X$ and $Y$, and any normalized state $\psi$, we have
\begin{equation}\label{general}
\Delta X\,\Delta Y\ \ge\
\frac 12\ \left|\,\langle\,\psi,\,[X,\,Y]\,\psi\,\rangle\,\right|.
\end{equation}
(One can take both sides as infinite if $\psi$ is not in the
appropriate domains.)

The formal argument for proving this is to note that the square of the norm
of
$$
\left\{\,\left(X\,-\,\langle\psi,\,X\,\psi\rangle\right)\,+\,
i\,\lambda\,\left(Y\,-\,\langle\psi,\,Y\,\psi\rangle\right)\,\right\}\ \psi
$$
must be positive.
By explicit calculation, this leads to the inequality
$$
\Delta X^2\,+\,\lambda^2\,\Delta Y^2\ +\,\lambda\,
\langle\,\psi,\,i\,[X,\,Y]\,\psi\,\rangle\ \ge\ 0.
$$
The left hand side is minimized by taking
$\displaystyle
\lambda\,=\,-\,\frac{\langle\,\psi,\,i\,[X,\,Y]\,\psi\,\rangle}{2\,\Delta Y^2}
$.
Since the inequality is true for this value for $\lambda$, we see that
$$
\Delta X^2\ -\ \frac{\langle\,\psi,\,i\,[X,\,Y]\,\psi\,\rangle^2}{4\,\Delta Y^2}
\ \ge\ 0,
$$
and inequality (\ref{general}) follows immediately.

Since $[x,\,p]\,=\,i\,\hbar$, inequality (\ref{general})
implies inequality (\ref{Heisenberg}). We then
note that for any value of $\theta$, we have
\begin{eqnarray}\nonumber
[\alpha,\,\beta]&=&
[x\,\cos(\theta)\,+\,p\,\sin(\theta),\,-\,x\,\sin(\theta)\,+\,p\,\cos(\theta)]
\\[2mm]\nonumber
&=&\cos^2(\theta)\,[x,\,p]\ -\ \sin^2(\theta)\,[p,\,x]
\\[2mm]\nonumber
&=&i\,\hbar.
\end{eqnarray}
So, by another application of inequality (\ref{general}), we obtain
the uncertainty relation that for any normalized state $\psi$,
\begin{equation}\label{rotated}
\Delta\alpha\,\Delta\beta\ \ge\ \frac{\hbar}2.
\end{equation}
Our main result is that one actually has equality when
$\psi\,=\,\varphi_0(A,\,B,\,\hbar,\,a,\,\eta,\,x)$ and
$\displaystyle\theta\,=\,\frac 12\,\mbox{arctan}\,
\left(\,\frac{2\ \mbox{Im}\,(B\overline{A})}{|B|^2-|A|^2}\,\right)$.

\vskip 5mm
%%%%%%%%%%%%%%%%%%%%
% Section 4
%%%%%%%%%%%%%%%%%%%%
\section{One Dimensional Semiclassical Wave Packets}
The one dimensional semiclassical wave packets are most easily defined
by using raising and lowering operators \cite{raise}:\quad
The number $a\in\mathbb{R}$ denotes the mean position of the wave packets.
The number $\eta\in\mathbb{R}$ denotes the mean momentum. We assume
the semiclassical parameter $\hbar$ is positive, and we choose any two complex
numbers $A$ and $B$ that satisfy
$\overline{A}B\,+\,\overline{B}A\,=\,2$.
We then define
$$
\varphi_0(A,\,B,\,\hbar,\,a,\,\eta,\,x)\ =\ \pi^{-1/4}\,\hbar^{-1/4}\,A^{-1/2}\,
\exp\left\{\,-\,\frac{B\,(x-a)^2}{2\,A\,\hbar}\ +\ i\,\eta\,(x-a)/\hbar\right\}.
$$
(The square root $A^{-1/2}$ can take either sign. In applications, the sign
is determined by an initial choice and continuity in time.)
We note that this vector is normalized because
$\overline{A}B\,+\,\overline{B}A\,=\,2$.
Also, any complex, normalized Gaussian can be written this way.

In analogy with the usual harmonic oscillator, we define raising and lowering
operators by
$$
{\mathcal A}(A,\,B,\,\hbar,\,a,\,\eta)^*\ =\
\frac 1{\sqrt{2\,\hbar}}\ 
\left[\,\overline{B}\,(x-a)\,-\,i\,\overline{A}\,(p-\eta)\,\right]
$$
and
$$
{\mathcal A}(A,\,B,\,\hbar,\,a,\,\eta)\ =\
\frac 1{\sqrt{2\,\hbar}}\ 
\left[\,B\,(x-a)\,+\,i\,A\,(p-\eta)\,\right].
$$
Using the raising operator inductively, starting from $k=0$, we define
$$
\varphi_{k+1}(A,\,B,\,\hbar,\,a,\,\eta,\,x)\ =\ \frac 1{\sqrt{k+1}}\
{\mathcal A}(A,\,B,\,\hbar,\,a,\,\eta)^*\
\varphi_k(A,\,B,\,\hbar,\,a,\,\eta,\,x).
$$
Then, for the lowering operator, we also have
$$
{\mathcal A}(A,\,B,\,\hbar,\,a,\,\eta)\
\varphi_0(A,\,B,\,\hbar,\,a,\,\eta,\,x)\ =\ 0
$$
and
$$
{\mathcal A}(A,\,B,\,\hbar,\,a,\,\eta)\
\varphi_k(A,\,B,\,\hbar,\,a,\,\eta,\,x)\ =\ \sqrt{k}\
\varphi_{k-1}(A,\,B,\,\hbar,\,a,\,\eta,\,x),
$$
for $k=1,\,2,\,\ldots\,$.

By using the raising and lowering operators, it is quite easy \cite{raise} to prove
that $\{\,\varphi_k(A,\,B,\,\hbar,\,a,\,\eta,\,x)\,\}$ is an orthonormal basis
of $L^2(\mathbb{R},\,dx)$, and that in the state
$\varphi_k(A,\,B,\,\hbar,\,a,\,\eta,\,x)$,
$$
\Delta x\ =\ \left(\frac\hbar 2\right)^{1/2}\ |A|\ \sqrt{2k+1},
$$
and
$$
\Delta p\ =\ \left(\frac\hbar 2\right)^{1/2}\ |B|\ \sqrt{2k+1}.
$$
The proof we present below for $\varphi_0$ generalizes to show that with\\
$\displaystyle\theta\,=\,\frac 12\,\mbox{arctan}\,
\left(\,\frac{2\ \mbox{Im}\,(B\overline{A})}{|B|^2-|A|^2}\,\right)$,
in the state $\varphi_k(A,\,B,\,\hbar,\,a,\,\eta,\,x)$ we have
$$
\Delta\alpha\,\Delta\beta\ =\ \frac\hbar 2\ (2k+1).
$$

\vskip 5mm
Our proof of the rotated uncertainty relation (\ref{rotated}) will make use of
the raising and lowering operators to represent $(x-a)$ and $(p-\eta)$. From the
definitions above, one easily sees that
\begin{equation}\label{xrepresentation}
(x-a)\ =\ \sqrt{\frac\hbar 2}\ \left\{\,
A\,\mathcal{A}(A,\,B,\,\hbar,\,a,\,\eta)^*\ +\
\overline{A}\,\mathcal{A}(A,\,B,\,\hbar,\,a,\,\eta)\,\right\}
\end{equation}
and
\begin{equation}\label{prepresentation}
(p-\eta)\ =\ i\ \sqrt{\frac\hbar 2}\ \left\{\,
B\,\mathcal{A}(A,\,B,\,\hbar,\,a,\,\eta)^*\ -\
\overline{B}\,\mathcal{A}(A,\,B,\,\hbar,\,a,\,\eta)\,\right\}.
\end{equation}

\vskip 5mm
%%%%%%%%%%%%%%%%%%%%
% Section 5
%%%%%%%%%%%%%%%%%%%%
\section{The Rotated Uncertainty Product}\label{MainSection}
We begin with a technical lemma that is proved by
simple calculation.

\begin{lemma}\label{prelim}
If $A(t)$ and $B(t)$ satisfy
\begin{eqnarray}\label{Aeqn}
\dot{A}(t)&=&i\,B(t)\qquad\mbox{and}
\\[2mm]\label{Beqn}
\dot{B}(t)&=&i\,A(t),
\end{eqnarray}
then we have the following time derivatives:
\begin{eqnarray}\label{fDeriv}
\frac{d\phantom{i}}{dt}\,
\left(|A(t)|^2\,|B(t)|^2\right)&=&
2\,\left(|A(t)|^2-|B(t)|^2\right)\ \mbox{Im}\,(B(t)\overline{A(t)}),
\\[2mm]\label{ImDeriv}
\frac{d\phantom{i}}{dt}\,
\left(\mbox{Im}\,(B(t)\overline{A(t)})\right)&=&|A(t)|^2-|B(t)|^2,
\qquad\mbox{and}
\\[2mm]\label{dogDeriv}
\frac{d\phantom{i}}{dt}\,
\left(|A(t)|^2-|B(t)|^2\right)&=&-\,4\,\mbox{Im}\,(B(t)\overline{A(t)}).
\end{eqnarray}
\end{lemma}

We now state and prove our main result:
\begin{theorem}\label{thm:main}
Let
$A,\,B,\,\hbar,\,a,\,\mbox{and}\ \eta$
be any allowed values of the parameters.
If we choose
$$
\theta\,=\,\frac 12\,\mbox{arctan}\,
\left(\,\frac{2\ \mbox{Im}\,(B\overline{A})}{|B|^2-|A|^2}\,\right),
$$
then the state $\varphi_0(A,\,B,\,\hbar,\,a,\,\eta,\,x)$
minimizes the uncertainty
product for $\alpha$ and $\beta$. {\em I.e.},
$$
\Delta\alpha\,\Delta\beta\ =\ \frac\hbar 2.
$$
\end{theorem}

\begin{proof}
We can prove this by an explicit, but tedious calculation or by the following
more appealing argument.

First we note that the values of $a$ and $\eta$ are irrelevant,
so we can set them both to zero.
Second we note that clockwise rotations of phase space are 
generated by the standard, frequency 1, classical harmonic oscillator Hamiltonian
$(p^2+x^2)/2$.  If we propagate with its quantum analog, the state
$\varphi_0(A(0),\,B(0),\,\hbar,\,0,\,0,\,x)$ evolves to a new Gaussian
$\varphi_0(A(t),\,B(t),\,\hbar,\,0,\,0,\,x)$, where \cite{raise}
$A(t)$ and $B(t)$ are given in equations (\ref{ODEs}) below.

Applying the ``counterclockwise'' rotation to the quantum operators
$x$ and $p$ to obtain $\alpha$ and $\beta$ is equivalent to keeping the
original operators $x$ and $p$, but rotating the state $\varphi_0$ in the
clockwise direction.
Thus, proving the theorem is equivalent to showing that
$\Delta x\,\Delta p\,=\,\hbar/2$  in the state
$\varphi_0(A(\theta),\,B(\theta),\,\hbar,\,0,\,0,\,x)$.

Since $\displaystyle\Delta x\,\Delta p\,=\,\frac\hbar 2\,|A(\theta)|\,|B(\theta)|$,
it suffices to show that $|A(\theta)|\,|B(\theta)|\,=\,1$.

We can find the minimum of $|A(t)|\,|B(t)|$ by setting the derivative of\\
$f(t)\,=\,|A(t)|^2\,|B(t)|^2$ to zero. By formula (\ref{fDeriv}), this requires
$$
\dot{f}(t)\ =\ 2\,\left(|A(t)|^2-|B(t)|^2\right)\,
\mbox{Im}\,(B(t)\overline{A(t)})
\ =\ 0.
$$
Thus, we must have~ $\mbox{Im}\,(B(t)\overline{A(t)})\,=\,0$~ or~
$|A(t)|^2-|B(t)|^2\,=\,0$.

To obtain a contradiction, suppose a relative minimum occurs with\\
$\mbox{Im}\,(B(t)\overline{A(t)})\,\ne\,0$. Then we must have
$|A(t)|^2-|B(t)|^2\,=\,0$.

From (\ref{ImDeriv}) and (\ref{dogDeriv}),
we see that the second derivative of $f$ is
$$
\ddot{f}(t)\ =\ 2\,\left(|A(t)|^2-|B(t)|^2\right)^2\,-\,
8\,\left(\mbox{Im}\,(B(t)\overline{A(t)})\right)^2.
$$
From our assumptions above, this quantity is strictly negative, and we have
found a maximum of $f$ instead of a minimum.
Thus, at any minimum of $f$, we must have
\begin{equation}\label{goal}
\mbox{Im}\,(B(t)\overline{A(t)})\ =\ 0.
\end{equation}
Since $\mbox{Re}\,(B(t)\overline{A(t)})\ =\ 1$, this condition forces
$|A(t)|\,|B(t)|\,=\,1$, and hence our desired result
$\displaystyle\Delta x\,\Delta p\,=\,\frac\hbar 2$ .

We next note that
\begin{eqnarray}\nonumber
A(t)&=&A(0)\,\cos(t)\,+\,i\,B(0)\,\sin(t)
\\[2mm]\label{ODEs}
B(t)&=&i\,A(0)\,\sin(t)\,+\,B(0)\,\cos(t).
\end{eqnarray}
We obtain these relations by explicitly solving the linear system
of ordinary differential equations (2.16) of \cite{raise} for this easy
special case.

From these relations, we see that
\begin{eqnarray}\nonumber
&&\mbox{Im}\,(B(t)\,\overline{A(t)})
\\[2mm]\nonumber
&=&\mbox{Im}\,\left\{\,
(i\,A(0)\,\sin(t)\,+\,B(0)\,\cos(t)\right)
\left(\overline{A(0)}\,\cos(t)\,-\,i\,\overline{B(0)}\,\sin(t)\right)
\\[2mm]\nonumber
&=&\hspace{-1pt}\mbox{Im}\left\{
i\,(|A(0)|^2-|B(0)|^2)\cos(t)\sin(t)+B(0)\overline{A(0)}\cos^2(t)+
A(0)\overline{B(0)}\sin^2(t)\right\}
\\[2mm]\nonumber
&=&
(|A(0)|^2-|B(0)|^2)\,\cos(t)\,\sin(t)\,+\,
\mbox{Im}\,(B(0)\overline{A(0)})\,(\cos^2(t)\,-\,\sin^2(t)).
\end{eqnarray}

\vskip 5mm
So, equation (\ref{goal}) is equivalent to
$$
(|A(0)|^2-|B(0)|^2)\,\sin(2t)\,+\,
2\,\mbox{Im}\,(B(0)\overline{A(0)})\,\cos(2t)\ =\ 0,
$$
which is satisfied if we choose~ $t\,=\,\theta$.
This proves the theorem.
\end{proof}

\begin{remark}
Our comments about the choice of $\theta$ in Section \ref{angle} relied on
operators that were quadratic in the raising and lowering operators.
In the spirit of the proof of Theorem \ref{thm:main}, we can make similar
comments that just involve the lowering operator.

Let $U(t)$ denote the propagator for the standard, frequency 1, quantum
harmonic oscillator. Choosing $a=0$ and $\eta=0$, we have
$\mathcal{A}\,\varphi_0(A,\,B,\,\hbar,\,0,\,0,\,x)=0$.
Thus,
\begin{eqnarray}\nonumber
0&=&U(t)\,\mathcal{A}\,\varphi_0(A,\,B,\,\hbar,\,0,\,0,\,x)
\\[2mm]\nonumber
&=&U(t)\,\mathcal{A}\,U(t)^{-1}\,U(t)\,\varphi_0(A,\,B,\,\hbar,\,0,\,0,\,x)
\\[2mm]\nonumber
&=&
\left[B(x\,\cos(t)+p\,\sin(t))+iA(-x\,\sin(t)+p\,\cos(t))\right]
U(t)\,\varphi_0(A,\,B,\,\hbar,\,0,\,0,\,x)
\\[2mm]\nonumber
&=&
\left[(B\cos(t)-iA\sin(t))x+i(A\cos(t)-iB\sin(t))\,p\right]
U(t)\,\varphi_0(A,\,B,\,\hbar,\,0,\,0,\,x).
\\[2mm]\label{wow}&&
\end{eqnarray}
If we choose $t=\theta$, the coefficients of $x$ and $p$ inside the
square brackets have the same complex phase.
This is equivalent to
$$
\left[\,\gamma\,x\,+\,i\,\delta\,p\,\right]\,
U(t)\,\varphi_0(A,\,B,\,\hbar,\,0,\,0,\,x)\ =\ 0,
$$
where $\gamma$ and $\delta$ are real.
This implies $U(t)\,\varphi_0(A,\,B,\,\hbar,\,0,\,0,\,x)$ is a phase times
a real Gaussian, and that it consequently minimizes $\Delta x\,\Delta p$.
Undoing the rotation of phase space shows that
$\varphi_0(A,\,B,\,\hbar,\,0,\,0,\,x)$ minimizes the uncertainty product of
$\Delta\alpha\,\Delta\beta$.

Requiring the coefficients in (\ref{wow}) to have the same phases is
equivalent to solving (\ref{goal}) and equivalent
to solving the eigenvalue problem of Section \ref{angle}.
\end{remark}

\begin{remark}
If one is not interested in the value of $\theta$, but only its existence, the
portion of the proof after equation (5.6) can be replaced by the following:
\\[4mm]
It suffices to prove that equation (5.6) is satisfied for some $t$.
All rotations of $\mathbb{R}^2$ are symplectic, and the rotation group $SO(2)$ is compact. So, there exists a value of $t$ at which
$g(t)=\left(\mbox{Im}\,(B(t)\overline{A(t)})\right)^2$ takes its minimum.
At that minimum,\\ $\dot{g}(t)=2\,g(t)\,(|A(t)|^2-|B(t)|^2)$ must be zero.
So, $g(t)=0$ or $(|A(t)|^2-|B(t)|^2)=0$.
If $g(t)\ne 0$, we must have
$(|A(t)|^2-|B(t)|^2)=0$, and in that case,
it follows from explicit calculation that
$\ddot{g}(t)\,=\,-\,8\,g(t)^2\,<\,0$.
This cannot happen at a minimum, so we conclude that $g(t)$ must be zero.
\end{remark}

\vskip 1cm
\bibliographystyle{amsalpha}

\end{document}